\documentclass[10pt, latterpaper, oneside, onecolumn ]{article}   	%
\usepackage[utf8]{inputenc}
\usepackage{geometry}                		
\usepackage{tikz}
\usetikzlibrary{positioning}
\usetikzlibrary{arrows}
\usepackage{amsmath,amssymb}
\usepackage{graphicx}
\usepackage{dsfont}				
\usepackage{bbm}
\usepackage{calc,pgf,xcolor}
\usepackage{nicefrac}
\usepackage{enumerate}
\usepackage{setspace}

\usepackage{amssymb}
\usepackage{caption}
\usepackage{subcaption}
\usepackage{amsthm}
\usepackage{float}
\usepackage{xcolor}

\newtheorem{theorem}{Theorem}
\newtheorem{proposition}{Proposition}
\newtheorem{conjecture}{Conjecture}
\newtheorem{corollary}{Corollary}
\newtheorem{lemma}{Lemma}

\theoremstyle{remark}
\newtheorem{example}{Example}
\newtheorem{remark}{Remark}
\def\m{\mathcal}
\def\dfn{\stackrel{\mathrm{def}}{=}}
\def\Expt{\mathbb{E}}
\def\guess{\mathsf{G}}
\def\ind{\mathds{1}}
\def\mid{\,|\,}

\newcommand{\ord}{\mathsf{ORD}}
\newcommand{\cut}{\mathsf{CUT}}
\newcommand{\maxcut}{\mathsf{MAXCUT}}
\newcommand{\opt}{\mathsf{opt}}
\newcommand{\zz}{\mathsf{ZZ}}

\DeclareMathOperator{\argmax}{argmax}

\title{Reducing Guesswork via an Unreliable Oracle}
\author{Amir~Burin and Ofer~Shayevitz \thanks{This work has been supported by an ERC grant no. 639573, a CIG grant no. 631983, and an ISF grant no. 1367/14. The authors are with the Department of Electrical Engineering - Systems, Tel Aviv University, Tel Aviv, Israel (emails: amirburin@gmail.com, ofersha@eng.tau.ac.il).}}
\date{}

\begin{document}

\allowdisplaybreaks 	
	\maketitle
	
	\begin{abstract}
		Alice holds an random variable $X$, and Bob is trying to guess its value by asking questions of the form ``is $X=x$?''. Alice answers truthfully and the game terminates once Bob guesses correctly. Before the game begins, Bob is allowed to reach out to an oracle, Carole, and ask her any yes/no question, i.e., a question of the form ``is $X\in A$?''. Carole is known to lie with a given probability $p$. What should Bob ask Carole if he would like to minimize his expected guessing time? When Carole is always truthful ($p=0$), it is not difficult to check that Bob should order the symbol probabilities in descending order, and ask Carole whether the index of $X$ w.r.t this order is even or odd. We show that this strategy is almost optimal for any lying probability $p$, up to a small additive constant upper bounded by a $1/4$. We discuss a connection to the cutoff rate of the BSC with feedback. 
	\end{abstract}
	
	\section{Introduction and Main Result}
	In the classical guessing game introduced and studied by Massey~\cite{massey1994guessing}, Alice is in possession of a discrete random variable (r.v.) $X$, and Bob would like to guess its value as quickly as possible. He is allowed to guess one symbol at a time, namely to ask Alice questions of the form ``is $X=x$?''; Alice truthfully answers Bob's guesses with ``yes'' or ``no'', and the game terminates when Bob guesses correctly. Bob's optimal strategy, in the sense of minimizing his expected guessing time, is to guess the symbols in decreasing order of probability~\cite{massey1994guessing}. Suppose now that before the game begins, Bob can reach out to an Oracle, Carole, and ask her a yes/no question of his choosing regarding $X$. What is the best question for him to ask in order to minimize his expected guessing time? Ordering the probability distribution of $X$ in descending order, it is not difficult to show that Bob is better off using the \textit{zigzag} query, i.e., asking whether $X$ has an even index or an odd index w.r.t. this order (there are other equally good queries). But, suppose now that Carole lies with some known probability $p$. What should Bob ask in this case? This seemingly simple problem turns out to be nontrivial. In this paper, we show that the zigzag query is almost optimal, in the sense of minimizing Bob's expected guessing time up to an additive constant of $\frac{\left|1-2p\,\right|}{4}$, independent of the cardinality of $X$. This is done by formulating the problem as that of finding a maximum cut in a certain weighted graph, and bounding the weight of the maximum cut via quadratic relaxation using special properties of the graph. We conjecture that the zigzag query is in fact exactly optimal for any $p$.
	
	Admittedly, on an intuitive level the optimality of the zigzag query may feel almost trivial. In order to get a glimpse of why this problem is not as easy as it might seem at a first glance, consider a simplified version where the only questions Bob is allowed to ask Carole are of the form ``is $X=x$?''. Namely, the question he can ask Carole (and get a noisy answer to) is also of the guessing form used in the game with Alice. What should Bob ask? It is perhaps initially tempting to think Bob is better off asking ``is $X=x_{\max}$?'', where $x_{\max}$ is the symbol with the maximal probability. This is of course correct when $p=0$, but is not true in general. To see why, consider an extreme case where $x_{\max}$ has probability very close to one, and where $p$ is sufficiently close to $1/2$. In this case, the posterior distribution of $X$ after receiving a very noisy answer to the question ``is $X=x_{\max}$?'', has exactly the same order as the prior distribution. In other words, Bob already knew that $x_{\max}$  is the most likely symbol, and Carole's answer is too noisy to change this belief. Thus, the optimal guessing strategy before and after receiving the answer remains the same, and therefore so does the expected guessing time. However, if for example there are two other symbols of equal probability, which a priori have no preferable order among themselves, then asking Carole about any of them would be beneficial for Bob as it would determine a preferable order. 
	
	What is the optimal ``is $X=x$?''  question then? The answer is somewhat counterintuitive. As we later demonstrate in Example ~\ref{exmpl:simple_problem_any_member_can_be_best_query}, for any $p>0$ and positive integers $k\leq N$, it is not difficult to show that there exists an r.v. $X$ of cardinality $N$ such that asking whether $X$ equals the $k$th largest symbol is the optimal question. The underlying reason yet again is that the probability of a symbol is not so important; what matters the most is whether the order of the symbol changes between the prior distribution of $X$ and the posterior distribution of $X$ given the answer, and how many other symbols it ``passes'' on its way up or down the probability order. This reveals the combinatorial nature of the problem, which arguably is what makes it more difficult. 
	
	We proceed to formally define the problem. Let $X$ be an r.v. taking values in a finite alphabet $\m{X}$, where without loss of generality will be assumed throughout to be $\m{X} = \{1,\ldots,N\}$, and let $P_X$ denote its probability mass function. A {\em guessing strategy} for Bob is any bijective function $g:\m{X}\to \{1,\ldots,N\}$, determining the order of guessing. Now, define $\guess(X)$ to be the minimal expected time required for Bob to correctly guess the value of $X$, i.e.,\footnote{Note that in the literature, $G(\cdot)$ typically denotes a guessing function (strategy), and $\Expt (G(X))$ is used as the expected guessing time.} 
	\begin{align}
		\guess(X) \dfn \min_{g} \Expt{\left(g(X)\right)}. 
	\end{align}
	where the minimum is taken over all guessing strategies. Let $\ord_X:\m{X}\to\{1,\ldots,N\}$ be the \textit{order function} of $X$, namely where $\ord_X(x)$ is the index associated with the probability $P_X(x)$ when the probabilities are ordered in a descending order, and where ties are resolved arbitrarily, say by taking the symbol with the smaller index to have a smaller order. It was observed by Massey~\cite{massey1994guessing} that $\ord_X$ is Bob's best guessing strategy, i.e., 
	\begin{align}
	\guess(X) = \Expt{\left(\ord_X(X)\right)}. 
	\end{align}
	This follows by noting that for any guessing strategy $g$ for which $P_X(j) \geq P_X(k)$ but $g(j) > g(k)$, replacing the guessing order of $j$ and $k$ cannot increase the expected guessing time.

We note in passing the following properties of $\guess(X)$. The proofs are relegated to the appendix.
\begin{proposition}\label{lem:guessing_time_prop}
Let $X,X'$ take values in the same finite alphabet, and suppose that the transition probability matrix $P_{X'|X}(x'\mid x)$ is doubly stochastic. Then $\guess(X)\leq \guess(X')$. 
\end{proposition}
\begin{corollary}\label{cor:guess_simple_bounds}
  $1\leq \guess(X) \leq (N+1)\slash 2$. The lower bound is attained if and only if $X$ is deterministic, and the upper bound is attained if and only if $X$ is uniform. 
\end{corollary}

The original guessing game has been extended by Arikan~\cite{arikan1996inequality} to the case where Bob has side information, in the form of another r.v. $Y$ over some alphabet $\m{Y}$, such that $(X,Y)\sim P_{XY}$ are jointly distributed. In this case, a {\em conditional guessing strategy} for Bob is a function $g:\m{X}\times \m{Y}\to \{1,\ldots,N\}$, with the property that $g(\cdot,y)$ is a bijection for any $y\in \m{Y}$, determining the order of guessing given that $Y=y$. We can similarly define $\guess(X\mid Y)$ to be the minimal expected time required for Bob to correctly guess the value of $X$ given that he knows $Y$, i.e.,   
\begin{align}
	\guess(X\mid Y) \dfn \min_g \Expt\left(g(X,Y)\right),
\end{align}
where the minimum is taken over all conditional guessing strategies. Let the \emph{conditional order function} $\ord_{X\mid Y}(x\mid y)$ be the order function pertaining to the distribution $P_{X|Y}(\cdot\mid y)$. It was observed by Arikan~\cite{arikan1996inequality} that $\ord_{X\mid Y}$ is Bob's best conditional guessing strategy, i.e., 
	\begin{align}
	\guess(X\mid Y) = \Expt{\left(\ord_{X|Y}(X\mid Y)\right)}. 
	\end{align}	
	This follows similarly by noting that for any conditional guessing strategy $g$ where $P_{X|Y}(j\mid y) \geq P_{X|Y}(k\mid y)$ but $g(j\mid y) > g(k\mid y)$, replacing the guessing order of $j$ and $k$ given $Y=y$ cannot increase the conditional expected guessing time.

	In this paper, we are concerned with side information that is actively obtained by asking a binary question and getting a noisy answer. A binary question corresponds to a {\em partition} of the alphabet into $\m{X} = A\cup \bar{A}$ for some $A\subseteq \mathcal{X}$. In the sequel, we informally refer to the set $A$ itself as the partition. Let 
	\begin{align}
	Y_A \dfn \ind(X\in A) \oplus V,
	\end{align}
	where $V\sim\text{Ber}(p)$ is independent of $X$ for some given $p$, and let 
	\begin{align}
	\guess_A(X) \dfn \guess(X\mid Y_A).
	\end{align}
	The quantity $\guess_A(X)$ is Bob's minimal expected guessing time of $X$ after asking Carole the binary question pertaining to $A$, and receiving an answer that is incorrect with probability $p$. We are interested in studying the best possible question / partition, i.e., to characterize 
	\begin{align}
	\guess_{\opt}(X)\dfn  \min_{A\subseteq\m{X}} \guess_A(X), 
	\end{align}
	as well as the question/partition that attains the minimum. From this point on, and without loss of generality, we assume that $p_k \dfn P_X(k)$ are non-increasing  $p_1\geq p_2\geq \cdots \geq p_N$. We define the \textit{zigzag partition}: 
	\begin{align}
	A_{\zz} \dfn \{k : \text{$k$ is odd}\}, 
	\end{align}
	and write $\guess_{\zz}(X)$ to denote $\guess_{A_{\zz}}(X)$. As we later show in Proposition~\ref{Proposition: Noiseless Case}, the zigzag partition is optimal when Carole is always truthful ($p=0$). More generally, we prove the following theorem. 
	\begin{theorem}[zigzag is almost optimal]\label{thrm:zz_is_almost_optimal}
		For any discrete r.v. $X$ and any $p\in[0,1]$
		\begin{align}
		\guess_{\zz}(X) \leq \guess_{\opt}(X) + \frac{\left|1-2p\,\right|}{4}.
		\end{align}
	\end{theorem}
	
	We conjecture the following. 
	\begin{conjecture}\label{conj}
		$\guess_{\zz}(X) =  \guess_{\opt}(X)$ for any discrete r.v. $X$ and any $p\in[0,1]$. 
	\end{conjecture}
	
	\subsection{Related Work}
	The classical problem of determining the value of a discrete r.v. $X$ by asking general binary questions is well studied in information theory and source coding, going back to Shannon~\cite{shannon1948mathematical} and Huffman~\cite{huffman1952method}. As is well known, this problem leads to the notion of Shannon entropy $H(X)$ as the essential fundamental limit for the minimal number of questions required on average to describe a single copy of $X$, and as the exact number of questions per instance (with high probability) required to describe i.i.d. copies of $X$ in the limit of multiple instances. More recently, Massey~\cite{massey1994guessing} introduced a different notion of r.v. complexity, corresponding to the minimal number of guesses required on average in order to determine the value of $X$, referred to here as $G(X)$. Massey related $G(X)$ to $H(X)$ by deriving a lower bound showing that the expected guessing time of $X$ grows at least exponentially with its Shannon entropy: 
	\begin{align}\label{eq:massey_lb}
		\guess(X) \geq 2^{H(X)}/4+1. 
	\end{align}
	This bound it tight within a $(4/e)$ multiplicative factor when $X$ is geometrically distributed. In a follow-up work, Arikan~\cite{arikan1996inequality} defined the notion of conditional guessing, and provided general lower and upper bounds for the $\rho$-th moment of the conditional guessing time of $X$ given $Y$, relating them to $H_{\frac{1}{1+\rho}}(X\mid Y)$, the Arimoto-R\'{e}nyi conditional entropy of order $\frac{1}{1+\rho}$. Arikan's upper bound without the conditioning was later tightened by Bozta{\c{s}}~\cite{boztas1997comments} for integer moments. In particular, when evaluated for a bivariate i.i.d. sequence $\{(X_k,Y_k)\}_{k=1}^n\stackrel{\mathrm{i.i.d}}{\sim} P_{XY}$ and $\rho=1$, Arikan's bounds imply that 
	\begin{align}\label{eq:arikan_asymp}
		\lim_{n\to\infty}\frac{1}{n}\log{\guess(X^n\mid Y^n)} = H_{1/2}(X\mid Y), 
	\end{align}
	with a similar result for general $\rho$. Continuing his previous work on the cutoff rate of single-user sequential decoding~\cite{arikan1988upper}, Arikan used the conditional guessing moment bounds to determine the cutoff rate of sequential decoding in multiple-access channels.  
	 
	Arikan's work has been extended in many directions. Arikan and Merhav~\cite{arikan1998guessing} considered the case of guessing a possibly continuous r.v, where Bob's guess is considered correct if it is close enough to the true value w.r.t some distortion measure. They derived a single letter variational expression for the  exponent of the guessing moment as a function of the distortion level. The same authors then extended the discussion to a joint source-channel coding setup with a guessing decoder~\cite{arikan1998joint}, and to the wiretap channel setting with a guessing wiretapper~\cite{merhav1999shannon}. Arikan and Bozta{\c{s}} considered a one-sided lying variation of the guessing game~\cite{arikan2002guessing}, where Alice lies with some probability when she rejects Bob's guesses, but never lies when he guesses correctly. Sundaresan~\cite{sundaresan2007guessing} studied the case of universal guessing, where the underlying distribution $P_X$ is only known to belong to some family of distributions, and determined the associated penalty (redundancy) in the guessing exponent incurred by this uncertainty. When specialized to the case of an i.i.d. distribution with an unknown marginal, his general results indicate that the redundancy term vanishes asymptotically, a fact that was already observed by Arikan and Merhav~\cite{arikan1998guessing}. Massey's guessing game has inspired a myriad of other works, tackling various other guessing setups and relations between guessing moments and entropy, see e.g. ~\cite{pfister2004renyi,malone2004guesswork,yamamoto2011channel,hanawal2011guessing,sason2018improved}, as well as discussing the implications and applications of guessing in cryptographic settings, see e.g. ~\cite{arikan2008guessing,malone2012investigating,christiansen2013guessing,bracher2015guessing,yona2017effect}, among many others.  

	The guessing game considered in this paper allows Bob, the guesser, to ask a single general binary question to which he obtains a possibly incorrect answer, before proceeding with symbol-by-symbol guessing. The first phase of our setup is thus reminiscent of another game, known as the {\em R\'{e}nyi-Ulam game}~\cite{renyi1961problem,ulam1978adventures}. In this game, Bob is allowed to ask Carole multiple general binary questions (adaptively), to which he obtains possibly incorrect answers, and his goal is to identify $X$. The classical version of this game is adversarial: The number of questions Bob can ask as well as the maximum number of lies Carole can tell are given, and Bob needs to find $X$ with certainty (so there is no need to assume a distribution on $X$, only the cardinality $N$ matters). The problem is to determine, for a given set of parameters, whether Bob has a strategy to always win the game, see~\cite{pelc2002searching} for a comprehensive survey. In his PhD thesis, Berlekamp~\cite{berlekamp1964block} studied the properties of winnable games from the equivalent perspective of error correction with noiseless feedback. Specifically, he provided bounds on the asymptotic version of this problem, where the cardinality of $X$ grows exponentially as $N=2^{nR}$ and the maximum number of lies grows linearly as $np$, and where $n$ is the total number of questions. Berlekamp's bounds together with a result by Zigangirov~\cite{zigangirov1976number} provide a complete characterization of the relation between $p$ and $R$, unlike the case where Bob needs to decide on his questions in advance, which is equivalent to the problem of finding the maximum growth rate of a binary error correcting code with minimum distance that scales linearly with the block length, a notorious open problem in coding theory.
	
	In our setup the lies are random; the version of the R\'{e}nyi-Ulam game in which Carole lies with probability $p$ and Bob needs to determine $X$ with high probability given Carole's answers, can essentially be thought of as the standard channel coding with noiseless feedback over a binary symmetric channel with crossover probability $p$~\cite{horstein1963sequential,shayevitz2011optimal}. Going back to our guessing game, if we allow Bob to ask Carole multiple questions before he starts guessing $X$, then our setup can in fact be viewed as channel coding with noiseless feedback where instead of a small error probability we are interested in a small expected guessing time at the decoder. This problem is closely related to that of the {\em cutoff rate} of the binary symmetric channel with feedback~\cite{gallager1968information,arikan1988upper}, which we briefly discuss later in this paper.  
	
	\subsection{Organization}
	In Section~\ref{Section: Noiseless Case} we analyze the simple noiseless case and show that zigzag (as well as some other partitions) is exactly optimal. Section~\ref{Section: Noisy Case} is devoted to the proof of our main result, showing that zigzag is almost optimal in the noisy case. In Section~\ref{sec:additional} we provide two additional results: We show that most non-zigzag partitions that are noiseless-optimal become strictly suboptimal with noise, and also that repeated zigzag partitions achieve the cutoff rate of the binary symmetric channel with feedback.   
		
\section{Noiseless Case}
	\label{Section: Noiseless Case}
	In this section we discuss the special case where Carole answers truthfully, i.e., $p=0$. In this simple setting, if Bob asks a question ``is $X\in A$?'', his posterior distribution is simply $P_X$ restricted to either $A$ or $\overline{A}$ according to Carole's answer, hence his optimal guessing strategy is going over the symbols in the relevant set in a descending order. We fully characterize the family of optimal questions in this simple case, and specifically show that the zigzag partition is optimal. We provide an explicit expression for $\guess_{\opt}(X)$. 

Below we assume whenever convenient and without loss of generality that $N$ is even (if this is not the case, we can append a zero probability symbol). We call $A$ a \textit{C-partition} if for any $i\in\{1,\ldots N/2\}$ it holds that either $2i-1\in A$ and $2i\in \overline{A}$, or $2i-1\in \overline{A}$ and $2i\in A$. In other words, when the symbols are divided into pairs in a descending order, the members of each pair are on the opposite side of the partition. When there are symbols with equal probabilities, we also call $A$ a C-partition if it can be trivially transformed into one by swapping between such symbols. Note in particular that the zigzag partition is a C-partition. 

We prove the following. 
	
	\begin{proposition}
		\label{Proposition: Noiseless Case}
		Let $p=0$. Then for any discrete r.v. $X$ of cardinality $N$,  
		\begin{equation}
		\guess_{\opt}(X) = \frac{1}{2}\left(\guess(X) + \sum_{i=1}^{\lceil N/2\rceil }p_{2i-1}\right), 
		\end{equation}
        which implies that 
		\begin{equation}\label{eq:g_opt_p=0}
		\frac{\guess(X)}{2} + \frac{1}{4}\leq \guess_{\opt}(X) \leq  \frac{\guess(X)}{2} + \frac{1}{2}.
		\end{equation}
                Moreover, $\guess_A(X) = \guess_{\opt}(X)$ if and only if $A$ is a C-partition. In particular, the zigzag partition is optimal. 
	\end{proposition}
	
\begin{proof}
The reason for the optimality of a C-partition is quite intuitive. Following Carole's answer, Bob's first guess will be correct if $X$ is equal to the maximal probability symbol on either side of the partition. To maximize this probability, the symbols $1$ and $2$ better be on opposite sides of the partition. Similarly, Bob's second guess will be correct if $X$ is equal to the second largest symbol on either side of the partition, hence he should place symbols $3$ and $4$ on opposite sides as well. Continuing this argument, it can be observed that we are getting a C-partition, and that it is an optimal one. 

To make this precise, let $A$ be some partition and assume without loss of generality that $|A|\geq N/2$. Recall that Carole's answer in this noiseless setting is  $Y_A = \ind(X\in A)$. The minimal expected guessing time associated with $A$ is  
\begin{align}
\guess_A(X) &= \Expt{\left(\ord_{X|Y_A}(X\mid Y_A)\right)} \\ 
 & =  \sum_{i=1}^Np_i \cdot \ord_{X|Y_A}(x_i\mid \ind(x_i\in A))\\ 
& = \sum_{k=1}^{|A|} k \left(\sum_{i\in A}p_i\ind(\ord_{X|Y_A}(x_i\mid 1)=k) + \sum_{i\in \overline{A}}p_i\ind(\ord_{X|Y_A}(x_i\mid 0)=k)\right)\label{eq:ordA} \\ 
& \geq \sum_{k=1}^{N/2} k (p_{2k-1}+p_{2k}) \label{eq:ordB}\\
& = \frac{1}{2}\sum_{k=1}^{N/2} \left((2k-1)p_{2k-1}+2kp_{2k} + p_{2k-1}\right)\\
& = \frac{1}{2}\left(\guess(X) + \sum_{k=1}^{N/2}p_{2k-1}\right).
\end{align}

In~\eqref{eq:ordA} we rearrange the summation and count on orders instead of symbols. The inequality~\eqref{eq:ordB} holds by virtue of the fact that each $k$ is multiplied by the sum of a distinct pair of symbol probabilities; the best possible assignment is to associate lower $k$ values with as much probability mass as possible. We now observe that the inequality~\eqref{eq:ordB} is tight if and only if $A$ is a C-partition, completing the proof. 

\end{proof}

\begin{example}[\emph{Multiple questions}]\label{ex:guessing_ent}
Suppose that before starting to guess $X$, Bob can ask Carole multiple binary questions. How many questions does he need in order for his guessing time to reduce to $1+\delta$, for some $0\leq \delta \leq \guess(X)-1$? And what is the best strategy? Asking $k$ questions (either sequentially depending on previous answers, or in a batch, since all answers are correct) is equivalent to partitioning the alphabet into $2^k$ disjoint subsets and asking Carole to point out the correct one. Following similar steps as done above, the best such partition is obtained by $2^k$-ary zigzag, i.e., putting symbol $i$ in subset $(i-1) \mod 2^k$. Moreover, this can be achieved sequentially by using simple zigzag queries $k$ times (adaptively). Let $\guess_{\opt}^{(k)}(X)$ be Bob's minimal expected guessing time after asking $k$ questions. Appealing to~\eqref{eq:g_opt_p=0}, we have that
\begin{align}
2^{-k}(\guess(X) -1/2) + 1/2 \leq  \guess_{\opt}^{(k)}(X) \leq 2^{-k}(\guess(X) -1) +1, 
\end{align}
and hence the number of required questions satisfies
\begin{align}\label{eq:noiseless_questions}
\log\left(\frac{\guess(X)-1/2}{\delta+1/2}\right)\leq   k\leq \log\left(\frac{\guess(X)-1}{\delta}\right).
\end{align}
We therefore conclude that Bob optimally requires $\log\guess(X)$ questions in order to reduce his expected guessing time to a constant $1+\delta$, up to an $O(\log(1/\delta))$ additive factor independent of $P_X$. This can be compared with the fact that when asking general binary questions, $H(X)$ questions are necessary and $H(X)+1$ questions are sufficient on average in order to completely learn $X$, e.g. using Huffman coding. Note that Massey's lower bound~\eqref{eq:massey_lb} indicates that  $\log\guess(X) \geq H(X)-2$, which is essentially tight for geometric distributions. In general however, the number of questions required to reduce the guessing time to a constant can be much larger than the entropy. For $X^n\stackrel{\mathrm{i.i.d}}{\sim} P_X$, Arikan's asymptotic result~\eqref{eq:arikan_asymp} in conjunction with~\eqref{eq:noiseless_questions} shows that $H_{1/2}(X)$ is roughly the number of guesses per instance that Bob optimally requires in order to reduce his guessing time of $X^n$ to a constant. This in general is strictly larger than $H(X)$, which is asymptotically the number of general binary questions per instance that are required in order to determine $X^n$ with probability approaching one. In Subsection~\ref{subsec:cutoff}, we generalize this discussion and determine the number of questions required to reduce the expected guessing time to sub-exponential in the noisy answers case, relating the result to the cutoff rate of the binary symmetric channel with feedback. 
\end{example}

\section{Noisy Case}
\label{Section: Noisy Case}
In this section we turn to the main focus of the paper, namely the case where Carole lies with probability $p$. We assume throughout and without loss of generality that $p_1\geq p_2\geq\cdots\geq p_N>0$, and that $0\leq p\leq 1/2$. Whereas in the noiseless case Carole's response simply eliminated all the symbols on the wrong side of the partition, in the noisy case Bob needs to carefully calculate the posterior distribution of $X$ given Carole's answer before guessing in the posterior descending order. This posterior distribution generally involves all the symbols, and can ``interlace'' symbols from opposite sides of the partition when ordered in descending order. We would like to find the optimal partition, namely the one such that the posterior order will result (on average w.r.t. Carole's lies) in the least expected guessing time.

To get a sense as to why this problem is nontrivial, consider the following simpler problem. Assume Bob is limited to choosing sets $A$ of cardinality $|A|=1$, i.e., to ask a question of the form "Is $X$ the $k$th member?". What is the optimal question in this case? As the following example demonstrates, for any positive integers $N$, $k\leq N$ and any $0<p< \tfrac{1}{2}$, there exists an r.v. $X$ of cardinality $N$ for which "is $X$ the $k$th member?" is (strictly) the optimal question to ask. 

\begin{example}
	\label{exmpl:simple_problem_any_member_can_be_best_query}
	Fix any cardinality $N\in \mathbb{N}$. Let  $0< p < \tfrac{1}{2}$ and $k\in\{1,2,...,N-1\}$. Pick some numbers $\alpha,\beta$ such that $0<\alpha<p<\beta<\tfrac{1}{2}$. Define the sequence $\{q_i\}_{i=1}^N$ as follows. Set $q_1=1$, and for $i>1$ set 
        \begin{align}
          q_i = \left\{\begin{array}{cr}\frac{\alpha}{1-\alpha}\,q_{i-1} & i\neq k+1 \\ \frac{\beta}{1-\beta}\,q_{i-1} & i=k+1\end{array}\right.. 
        \end{align}
Now, consider an r.v. $X$ with distribution $p_i\dfn q_i/\sum_jq_j$. By construction, if Bob asks about any symbol other than $k$, then the posterior distribution given any answer by Carole has the exact same order as the prior, and hence the minimal expected guessing time remains the same. However, if Bob asks about $k$, then in case Carole says ``yes'' then $k$ strictly moves up at least one spot in the posterior order; hence, the minimal expected guessing time is reduced. This example highlights the fact that the most important feature of good partitions is not so much the exact symbol probabilities, but rather which symbol can ``pass'' which in the posterior probability order. 
\end{example}

Before proceeding, we note that it is sometimes technically convenient to consider setups in which all the distributions under consideration do not have symbols with equal probabilities. Precisely, we say that $P_X$ is {\em non-degenerate} (for a fixed $p$) if for any partition $A\subseteq \m{X}$ and answer $y\in\{0,1\}$, the symbol probabilities of the associated conditional distribution $P_{X|Y_A}(\cdot\mid y)$  are all distinct. Note that this includes in particular $P_X$ itself, which can obtained by choosing the empty partition $A=\emptyset$. The following lemma shows that the set of non-degenerate distributions is dense in the probability simplex. The proof is relegated to the appendix. 
\begin{lemma}\label{lem:non-degenerate}
	For any $P_X$ and any $\varepsilon > 0$, there exists a non-degenerate $Q_X$ such that $\|P_X - Q_X\|_1 < \varepsilon$. 
\end{lemma}
As a direct consequence, we have the following corollary. 
\begin{corollary}\label{cor:non-degenerate}
	Let $g$ be a continuous real-valued function over the $N$-dimensional probability simplex equipped with the $L_1$ metric.  Then $g(P_X)$ can be approximated arbitrarily well by non-degenerate distributions.   
\end{corollary}
As we shall see, the functions considered in the proofs below are all continuous in $P_X$, and hence assuming non-degeneracy will incur no loss of generality. 

\subsection{Graph Theoretic Formulation}
In this subsection, we reformulate our problem in graph theoretic terms, by defining a weighted graph whose maximal cut is equal to the maximal possible reduction in expected guessing time. Specifically, let $\m{G}_X$ be a simple undirected weighted graph with vertex set $\m{X}=\{1,\ldots,N\}$, where the weight of the edge connecting any two vertices $i\neq j$ is given by 
\begin{align}\label{eq:weights_def}
w_{i,j}= \big{|}(1-p)\min\{p_i,p_j\} - p\max\{p_i,p_j\}\big{|}_+, 
\end{align}
where $|a|_+ \dfn \max\{0,a\}$. We will typically assume that $i<j$ and hence that $w_{ij} = |(1-p)p_j-pp_i|_+$. 

The cut weight associated with any partition $A\subseteq \m{X}$, i.e., the total weight of edges connecting $A$ and $\bar{A}$, is denoted by  
\begin{align}
\cut_A(\m{G}_X) \dfn \sum_{i\in A, j\not\in A} w_{i,j}, 
\end{align}
and the maximal cut weight is 
\begin{align}
\maxcut(\m{G}_X) \dfn \max_{A\subseteq \m{X}}\cut_A(\m{G}_X).
\end{align}
We have the following result. 
\begin{theorem}
	\label{thrm:graph_formulation}
	For any discrete r.v. $X$ and any $0\leq p\leq \tfrac{1}{2}$, 
	\begin{align}
	\guess_A(X) = \guess(X) - \cut_A(\m{G}_X).
	\end{align}
	Specifically, 
	\begin{align}
	\guess_\opt(X) = \guess(X) - \maxcut(\m{G}_X).
	\end{align}
\end{theorem}

Before we prove the theorem, a few remarks are in order. 

\begin{remark}
	Note that for $i<j$, the vertices $i$ and $j$ are not connected in $\m{G}_X$ (i.e., $w_{i,j} = 0$) if and only if $\frac{p_i}{p_j} \geq \frac{1-p}{p}$. In other words, $i<j$ are not connected if and only if their respective probabilities $p_i > p_j$ are sufficiently far so that regardless of the question asked and the answer received, they will never pass each other in the posterior order.
\end{remark}

\begin{remark}[\emph{almost noiseless case}]
	In the noiseless case $p=0$, the graph $\m{G}_X$ is fully connected and the weights are given by $w_{i,j} = \min\{p_i,p_j\}$. It can be verified that the maximal cut is indeed attained by the zigzag partition, as expected.  More generally, if $p$ is small enough to render the graph fully connected (all positive weights), then again it can be readily shown that the zigzag partition achieves the maximal cut, and is hence optimal. For this to happen, it should hold that $\frac{1-p}{p} > \tfrac{\max p_i}{\min p_i} = \tfrac{p_1}{p_N}$.
\end{remark}

\begin{remark}[\emph{very noisy case}]
	On the other extreme, if $p$ is sufficiently close to half then the graph becomes empty. In this case $\maxcut(\m{G}_X)=0$  and is achieved by any partition. In other words, when the answer is too noisy there may be no gain to be reaped by asking the question. Specifically, this happens when $\frac{1-p}{p} \leq \displaystyle{\min_i\tfrac{p_{i}}{p_{i+1}}}$, which corresponds to the case where the prior and posterior orders are always the same for any possible partition. 
\end{remark}

\begin{proof}[Proof of Theorem~\ref{thrm:graph_formulation}]
	Fix any partition $A$, and let $Y_A = \ind(X\in A)\oplus V$ be Carole's noisy answer to Bob's question. First, let us consider by how much the order of a symbol $i\in\m{X}$ changes after Carole answers the question. Applying Bayes law, the posterior distribution of $X$ given $Y_A$ is given by 
	\begin{align}
	P_{X|Y_A}(i\mid y) &= \frac{P_{Y_A|X}(y\mid i)\cdot P_X(i)}{P_{Y_A}(y)}\\
	&=\label{eq:ord_posterior} \frac{1}{\frac{1}{2}(1+(1-2y)(1-2p)(1-2p_A))} \cdot\left\{\begin{array}{lr} (1-p)p_i & y = \ind(i\in A) 
	\\ pp_i & y = \ind(i\not\in A)\end{array}\right.,
	\end{align}
	where $p_A\dfn \sum_{j\in A}p_j$. 		
	
	For the purpose of technical convenience and to avoid edge cases, we assume below that $P_X$ is non-degenerate, which means that $\{p_i\}$ are all distinct, and that $(1-p)\min\{p_i,p_j\} \neq p\max\{p_i,p_j\}$ for all pairs $i\neq j$. This incurs no loss of generality: Examining~\eqref{eq:ord_posterior}, we see that $P_{X|Y_A}(\cdot \mid y)$ is a continuous function of $P_X$ for any fixed choice of $A, y$ and $p$. Since by definition $G_A(X)$ is linear in $P_{X|Y_A}(\cdot \mid y)$, it is continuous in $P_X$ as well . Similarly, the edge weights of $\m{G}_X$ are continuous in $P_X$, thus so is $\cut_A(\m{G}_X)$ for any fixed $A$, and so is $\maxcut(\m{G}_X)$ being the maximum of a finite number of continuous functions. Hence, in light of Corollary~\ref{cor:non-degenerate}, $\guess(X), \guess_A(X), \guess_\opt(X), \cut_A(\m{G}_X), \maxcut(\m{G}_X)$ can all be approximated arbitrarily well by non-degenerate distributions.   
	
	We now note that each symbol retains its order w.r.t. all other members on the same side of the partition, i.e., for any $i,j\in A$ it holds that $\ord_X(i) > \ord_X(j)$ implies $\ord_{X|Y_A}(i\mid y) > \ord_{X|Y_A}(j \mid y)$ (and similarly for $i,j\in\bar{A}$). However, the order may not be preserved between symbols on the opposite side of the partition, and a symbol $i\in A$ may pass or be passed by members of $\bar{A}$. Precisely: Let $A(i)$ be equal to $A$ if $i\in A$ and to $\bar{A}$ otherwise. Write $i\sim j$ if $w_{i,j}>0$. Then eq.~\eqref{eq:ord_posterior} implies that
        \begin{align}          
		&\ord_X(i) - \ord_{X|Y_A}(i\mid \ind(i\in A)) \\
                &= \left|\left\{j\not\in A(i) :  p_j > p_i,  (1-p)p_i-pp_j > 0\right\}\right| \\
		\label{eq:ord1}& = \left|\left\{j\not\in A(i) :  p_j>p_i, j\sim i\right\}\right|,	 
	\end{align}
        and 
        \begin{align}          
	&\ord_X(i) - \ord_{X|Y_A}(i\mid \ind(i\not\in A)) \\
        &= -\left|\left\{j\not\in A(i) : p_j < p_i, (1-p)p_j - pp_i > 0\right\}\right| \\ 
	\label{eq:ord2}& = -\left|\left\{j\not\in A(i) :  p_j<p_i, j\sim i\right\}\right|, 
        \end{align}          
where we can use strict inequalities between the probabilities due to non-degeneracy. Now, recall that $V=Y_A\oplus \ind(X\in A) \sim \text{Ber}(p)$ is independent of $X$. Then
        \begin{align}
	&\guess(X) - \guess_A(X) 
	\\&= \Expt\left(\ord_X(X) - \ord_{X|Y_A}(X\mid Y_A)\right)
	\\&= \Expt_{V}\Expt_{X,Y_A\mid V}\left(\ord_X(X) - \ord_{X|Y_A}(X\mid Y_A) \mid V\right) 
	\\&= (1-p)\,\Expt\left(\ord_X(X) - \ord_{X|Y_A}(X\mid \ind(X\in A))\right) 
	\\&\quad + p\,\Expt\left(\ord_X(X) - \ord_{X|Y_A}(X\mid \ind(X\not\in A))\right) 
	\\\label{eq:use_ord12}&= (1-p)\sum_i p_i \cdot \left|\left\{j\not\in A(i) :  p_j>p_i, j\sim i\right\}\right|
	\\&\quad - p\sum_i p_i\cdot \left|\left\{j\not\in A(i) :  p_j<p_i, j\sim i\right\}\right|
	\\&= (1-p)\sum_i \sum_{\substack{j\not\in A(i)\\p_j>p_i\\j\sim i}} p_i - p\sum_i \sum_{\substack{j\not\in A(i)\\p_j<p_i\\j\sim i}} p_i
	\\&= \sum_i \sum_{\substack{j\not\in A(i)\\p_j>p_i\\j\sim i}} (1-p)\min\{p_i,p_j\} - \sum_i \sum_{\substack{j\not\in A(i)\\p_j<p_i\\j\sim i}} p\max\{p_i,p_j\}
	\\\label{eq:w_ord}&= \sum_{i\in A} \sum_{j\not\in A} w_{i,j}
	\\& = \cut_A(\m{G}_X).
      \end{align}
We have used~\eqref{eq:ord1} and ~\eqref{eq:ord2} in~\eqref{eq:use_ord12}, and~\eqref{eq:w_ord} follows from the weights definition by reordering the summations. 
\end{proof}

\subsection{Properties of $\m{G}_X$}
\label{Section: Noisy Case : Subsection: Properties of GX}
In this subsection we point out two basic properties of the graph $\m{G}_X$, \emph{monotonicity} and \emph{additivity}, which prove useful in the sequel. To that end, it is instructive to think of an edge $(i,j)$ as an interval $[i,j]$ (which is identified with $[j,i]$ if $i>j$). With this in mind, we say that two edges intersect or contain each other whenever their intervals do. For edges that intersect, we can define their union and intersection to be the edges that correspond to the union and intersection of their intervals, respectively. 

\begin{lemma}[Monotonicity]
	\label{Lemma: Monotonicity}
	The weight of an edge is monotonically non-increasing w.r.t. containment, i.e., if $[i,j]\subseteq [k,l]$ then $w_{i,j}\geq w_{k,\ell}$. Specifically, if an edge is disconnected (i.e., has zero weight) then so are all the edges that contain it. 
\end{lemma}

\begin{proof}
By definition~\eqref{eq:weights_def} we immediately have that $w_{i,j} = |(1-p)p_i-pp_j|_+ \geq  |(1-p)p_k-pp_\ell|_+$. 
\end{proof}


		
		
\begin{lemma}[Additivity]
The sum of weights of two intersecting edges is equal to the sum of weights of their union and intersection, provided that all four edge weights are nonzero. Namely, for any $i,k\leq j,\ell$ with $w_{i,j},w_{k,l},w_{i,l},w_{k,j}>0$,  
	\label{Lemma: Additive property of weights}
	\begin{equation}
	w_{i,j}+w_{k,l}=w_{i,l}+w_{k,j}.
	\end{equation}
\end{lemma}
\begin{proof}
By definition~\eqref{eq:weights_def} and the assumption of positive weights, we have
\begin{align}
w_{i,j}+w_{k,l}&=(1-p)p_i-pp_j
\\&\quad+(1-p)p_k-pp_l
\\&=(1-p)p_i-pp_l
\\&\quad+(1-p)p_k-pp_j
\\&= w_{i,l}+w_{k,j}.
\end{align}
\end{proof}

\subsection{A Weak Bound via a Greedy Algorithm}
We are now in a position to prove a weak version of the zigzag optimality. It is well known and easy to check that a cut chosen uniformly at random has an expected cut weight equal to at least half the maximal cut weight. More interestingly, it is also known that this procedure can be derandomized via the conditional probabilities technique~\cite{alon2004probabilistic}, to show that a greedy algorithm attains at least this average performance, i.e., at least half the maximal. It turns out that we can set up a greedy algorithm that always converges to the zigzag partition, and hence show that zigzag achieves at least half of the maximal possible guessing time reduction. 
\begin{theorem}[zigzag is greedy]\label{thrm:zz_weak_bound}
	The greedy max-cut algorithm applied to  $\m{G}_X$ by adding the vertices in descending order of probability, yields the zigzag partition (possibly with suitable tie-breaking). Moreover, 
        \begin{align}\label{eq:weak_zz_bound}
          \cut_\zz(\m{G}_X) \geq \tfrac{1}{2}\maxcut(\m{G}_X).
        \end{align}
\end{theorem}

\begin{proof}
We prove by induction on the cardinality $N$ of $X$. For $N=2$, zigzag is trivially the greedy (and also optimal) solution. Assume that for cardinality $N=k$, descending order greedy yields zigzag, and let $X$ be any r.v. with cardinality $N=k+1$. We can see that the induced subgraph of $\m{G}_X$ corresponding to the vertices $\{1,\ldots,k\}$ has the same weights as the graph associated with the distribution $\{p_i/\sum_{j=1}^kp_j\}_{i=1}^k$, up to the normalizing scaling factor. Thus, by our induction assumption, descending order greedy applied to $\{p_1,\ldots,p_k\}$ yields the zigzag partition over $\{1,\ldots,k\}$. We are therefore left with the assignment of the least likely symbol $p_{k+1}$, which can either continue the zigzag pattern or break it. Let us look at the difference between the weight added to the cut by the first option and that of the second option: 
\begin{align}
\sum_{j=1}^{\left\lfloor (k+1)/2\right\rfloor}&w_{k+2-2j,k+1} - \sum_{j=1}^{\left\lfloor k/2\right\rfloor}w_{k+1-2j,k+1}\\
  &=\sum_{j=1}^{\left\lfloor (k+1)/2\right\rfloor}(w_{k+2-2j,k+1} - w_{k+1-2j,k+1}) \\&\geq 0, \label{eq:zz_takes_half}
\end{align}
where weights are set to zero whenever the indices go out-of-bounds. The inequality above follows directly from the monotonicity and nonnegativity properties of the weights. Thus, continuing the zigzag pattern is always at least as good as breaking it, establishing the inductive step. 

The inequality~\eqref{eq:weak_zz_bound} now holds due to the aforementioned general result~\cite{alon2004probabilistic} indicating that greedy achieves at least half the maximal cut. But in fact, in our setup this can be proved directly by the same induction on the cardinality $N$. Inequality ~\eqref{eq:weak_zz_bound} trivially holds for $N=2$. Assuming it holds for $N=k$ and using the same rationale as above, we can see from~\eqref{eq:zz_takes_half} that adding $p_{k+1}$ in the zigzag position adds at least half of the total weight of all edges connected to vertex $k+1$ into the cut, thereby establishing the inductive step. 
\end{proof}

Combining Theorem~\ref{thrm:graph_formulation} and Theorem~\ref{thrm:zz_weak_bound}, we immediately obtain the following. 
\begin{corollary}
	For any discrete r.v. $X$ and any $p$, it holds that 
	\begin{align}
	\guess_{\zz}(X) \leq \tfrac{1}{2}\left(\guess(X) + \guess_{\opt}(X)\right).
	\end{align}
\end{corollary}
We have thus shown that the zigzag partition is at worst half-way between the best and the worst partitions in terms of guessing time. This weak bound can be loose by an additive factor that is linear in $\guess(X)$, e.g., in the noiseless case. Note that we have not used the additivity property when proving this bound; only monotonicity has been utilized. Next, we prove Theorem~\ref{thrm:zz_is_almost_optimal} which shows that the zigzag partition is optimal up to a small additive constant, that is independent of the distribution and cardinality of $X$. 

\subsection{Proof of Theorem \ref{thrm:zz_is_almost_optimal} (main result)}
To prove our main result, we represent the associated max-cut problem as an integer quadratic programming problem, and then relax the integer assumption to optimize over the reals. This is a standard approach in combinatorial optimization \cite{schrijver2002combinatorial}. Once this is done, it will remain to show that the resulting matrix is positive semidefinite, and to bound the relaxation loss. While the latter is a simple exercise, proving positive semidefiniteness requires a specialized manipulation tailored to our graph, that makes use of both its monotonicity and additivity properties. 

For any partition $A$, let us define the partition assignment vector $\mathbf{x}\in\{1,-1\}^n$ such that $x_i=\ind(i\in A)-\ind(i\in\bar{A})$, i.e. the $i$-th coordinate is set to $1$ if the member $i$ is in the set $A$, and to $-1$ otherwise. Using this notation, the associated cut weight can be expressed as 
\begin{align}\label{eq:cut_as_quad_form}
\cut_A(\m{G}_X)=\frac{1}{4}\sum_{i,j}\left(1-x_ix_j\right)w_{i,j},
\end{align}
where $w_{i,j}$ are the weights of the edges in $\m{G}_X$, given in~\eqref{eq:weights_def}. Note that although $w_{i,i}$ is not defined in the original context (there are no self loops), these weights can be chosen arbitrarily as they do not change the value of the cut in~\eqref{eq:cut_as_quad_form}. Thus, for the purpose of the optimization to follow, we find it convenient to naturally define these diagonal weights to satisfy~\eqref{eq:weights_def} as well, i.e., we set $w_{i,i} = (1-2p)p_i$. Also, for brevity of exposition, we naturally set out-of-bounds weights to zero, i.e., $w_{i,j}=0$ whenever either $i$ or $j$ are not in the range $1,\ldots,N$. Practicing some algebra, we obtain
\begin{align}\label{eq:SD_relaxation}
\cut_A(\m{G}_X)=\frac{1}{4}\left(\sum_{i,j}w_{i,j}\right)-\frac{1}{4}\mathbf{x^TWx},
\end{align}
where the symmetric matrix $\mathbf{W}$ has entries $w_{i,j}= \big{|}(1-p)\min\{p_i,p_j\} - p\max\{p_i,p_j\}\big{|}_+$ for all $i,j$. 

We now claim the following. 
\begin{lemma}\label{lemma:w_is_psd}
	The matrix $\mathbf{W}$ is positive semidefinite. 
\end{lemma}

Showing that $\mathbf{W}$ is positive semidefinite directly appears to be quite difficult, mainly due to the nonlinear nullifying operator $|\cdot|_+$, which renders the eigenvalues and the determinant of principal minors intractable. To circumvent this, we will perform a certain manipulation that makes use of the special monotonic/additive structure of $\mathbf{W}$. To that end, we first need a few simple lemmas. The proofs are relegated to the appendix. 

Recall that two square matrices $\mathbf{B}$ and $\mathbf{C}$ are called \emph{congruent} if there exists an invertible matrix $\mathbf{A}$ such that $\mathbf{C} = \mathbf{A}^T\mathbf{B}\mathbf{A}$. 
\begin{lemma}\label{lem:eqv_psd}
  Let $\mathbf{B},\mathbf{C}$ be two congruent matrices. Then $\mathbf{B}$ is positive semidefinite if and only if $\mathbf{C}$ is positive semidefinite. 
\end{lemma}

\begin{lemma}[Gershgorin's Disks~\cite{bhatia2013matrix}]\label{lem:gdisks}
	Let $\mathbf{B}$ be a square complex matrix with entries $b_{i,j}$. Define $R_i=\sum_{j\neq i}|b_{i,j}|$, and let $D_i\subset \mathbb{C}$ be a disk of radius $R_i$ centered at $b_{i,i}$. Then each eigenvalue of $\mathbf{B}$ lies in at least one $D_i$.
\end{lemma}

The matrix $\mathbf{B}$ is called \emph{diagonally dominant} if it is real-valued $|b_{i,i}| \geq R_i$, i.e., the  Gershgorin disks all lie in the right-hand-side of the complex plane. The following is an immediate consequence of Lemma~\ref{lem:gdisks}. 
\begin{corollary}\label{cor:diag_dom}
 Suppose $\mathbf{B}$ is a symmetric, real-valued, diagonally dominant matrix, with nonnegative diagonal entries. Then $\mathbf{B}$ is positive semidefinite. 
\end{corollary}

\begin{proof}[Proof of Lemma~\ref{lemma:w_is_psd}] 
Our strategy is to find an invertible matrix $\mathbf{A}$ such that $\mathbf{Q} = \mathbf{A}^T\mathbf{W}\mathbf{A}$ is diagonally dominant, for any distribution $P_X$ and any value of $0\leq p \leq 1/2$. In light of Lemma~\ref{lem:eqv_psd} and Corollary~\ref{cor:diag_dom}, this will conclude our proof. 

Consider the following matrix $\mathbf{A}$ with entries $a_{i,j}=\delta_{i,j}-\delta_{i,j+1}$, where $\delta_{i,j}$ is the Kronecker delta:
\begin{align}
\mathbf{A}=\begin{bmatrix}
1 & 0 & 0 & 0 & \dots & 0 & 0 & 0 \\
-1 & 1 & 0 & 0 & \dots & 0 & 0 & 0 \\
0 & -1 & 1 & 0 & \dots & 0 & 0 & 0 \\
0 & 0 & -1 & 1 & \dots & 0 & 0 & 0 \\
\vdots & \vdots & \vdots & \vdots & \ddots & \vdots & \vdots & \vdots \\
0 & 0 & 0 & 0 & \dots & -1 & 1 & 0 \\
0 & 0 & 0 & 0 & \dots & 0 & -1 & 1 \\
\end{bmatrix}.
\end{align}
Note that the change of basis induced by this matrix essentially corresponds to taking a ``discrete derivative'' -- this in some sense implies that working with the derivative of the partition assignment vector makes a more useful representation. It can be seen that $\mathbf{A}$ is invertible. Let us compute the matrix $\mathbf{Q}=\mathbf{A^TWA}$ and show that it is diagonally dominant. Multiplying from the right, we obtain 
\begin{align}
u_{i,j} &\dfn \left(\mathbf{WA}\right)_{i,j}\\
&=w_{i,j}-w_{i,j+1},
\end{align}
where recall we define $w_{i,N+1}\dfn 0$. Now multiplying from the left, we obtain
\begin{align}
q_{i,j}&=\left(\mathbf{A^TWA}\right)_{i,j}\\
&=u_{i,j}-u_{i+1,j}
\\&=w_{i,j}+w_{i+1,j+1}-w_{i,j+1}-w_{i+1,j}. \label{eq:Qmat_entries}
\end{align}

Using~\eqref{eq:Qmat_entries} and the properties of the matrix $\mathbf{W}$, let us now proceed to show that $\mathbf{Q}$ has nonnegative diagonal entries, and non-positive off-diagonal entries. Monotonicity implies that $w_{i,i}+w_{i+1,i+1}\geq w_{i+1,i}+w_{i,i+1}=2w_{i,i+1}$, and therefore by~\eqref{eq:Qmat_entries} we have that $q_{i,i}\geq 0$. Now consider off-diagonal entries $q_{i,j}$ and assume without loss of generality that $i<j$. If all the edge weights in the expression~\eqref{eq:Qmat_entries} for  $q_{i,j}$ exist (i.e., for $p$ small enough), then by additivity $q_{i,j}=0$. As $p$ is increased, monotonicity implies that the first edge weight to become zero is $w_{i,j+1}$, at which point we have 
	\begin{align}
	q_{i,j}&=w_{i,j}+w_{i+1,j+1}-w_{i+1,j}\\
&=(1-p)p_{j+1}-pp_i\\
&\leq 0, 
	\end{align}
where the last inequality follows from $w_{i,j+1}=\left|(1-p)p_{j+1}-pp_i\right|_+$ and our assumption that $w_{i,j+1}=0$. When $p$ is further increased, the next edge weight to become zero is either $w_{i,j}$ or $w_{i+1,j+1}$. The following two edges to become zero are $w_{i,j},w_{i+1,j+1}$, which leaves us with $q_{i,j}$ that equals either $w_{i+1,j+1}-w_{i+1,j}$, or $w_{i,j}-w_{i+1,j}$, or $-w_{i+1,j}$, all of which are non-positive by monotonicity and nonnegativity of the weights. Finally, when $p$ is large enough, all the participating weights are zero and $q_{i,j}=0$.

We are now in a position to show that $\mathbf{Q}$ is diagonally dominant:
\begin{align}
q_{i,i}-\sum_{j=1,j\neq i}^{N}| q_{i,j} | &= \sum_{j=1}^{N}q_{i,j} \label{eq:using_signs}
\\& = \sum_{j=1}^{N}\left( w_{i,j}+w_{i+1,j+1}-w_{i,j+1}-w_{i+1,j} \right) \label{eq:expanding_q}
\\ &= w_{i,1}-w_{i+1,1}+w_{i+1,N}-w_{i,N}\\
&\geq 0, \label{eq:finalizing}
\end{align} 
where~\eqref{eq:using_signs} follows since the off-diagonal entries of $\mathbf{Q}$ are non-positive, ~\eqref{eq:expanding_q} holds by virtue of~\eqref{eq:Qmat_entries}, and~\eqref{eq:finalizing} follows from monotonicity. Combined with the fact that the diagonal entries of $\mathbf{Q}$ are non-negative and appealing to Corollary~\ref{cor:diag_dom}, we conclude that $\mathbf{Q}$ is positive semidefinite. 
\end{proof}

Returning to~\eqref{eq:SD_relaxation}, and noting that Lemma~\ref{lemma:w_is_psd} indicates that $\mathbf{x^TWx}\geq 0$ for any $\mathbf{x}$, we have that
\begin{align}\label{eq:maxcut_ub}
\maxcut(\m{G}_X)\leq\frac{1}{4}\sum_{i,j}w_{i,j}.
\end{align}
Let us relate this upper bound to $\cut_\zz(\m{G}_X)$. To that end, note that 
\begin{align}\label{eq:zz_formula}
\cut_\zz(\m{G}_X)=\sum_{i=1}^{N-1}\sum_{k=1}^{\left\lfloor N/2\right\rfloor}w_{i,i+2k-1}.
\end{align}
Observe also that by monotonicity,  $w_{i,i+2k-1}\geq w_{i,i+2k}$ and therefore 
\begin{align}\label{eq:mono_odd_even}
2w_{i,i+2k-1}\geq w_{i,i+2k-1} + w_{i,i+2k}
\end{align}
for any $k\geq 1$ (note that this bounds holds vacuously in case $2k>N$). Now we can write 
\begin{align}
\maxcut(\m{G}_X)&\leq\frac{1}{4}\sum_{i,j}w_{i,j} \\ 
&=\frac{1}{4}\sum_{i=1}^{N}w_{i,i}+\frac{1}{2}\sum_{i=1}^{N-1}\sum_{j=i+1}^{N}w_{i,j} \\ 
&=\frac{1}{4}\sum_{i=1}^{N}(1-2p)p_i+\frac{1}{2}\sum_{i=1}^{N-1}\sum_{k=1}^{\left\lfloor N/2\right\rfloor}(w_{i,i+2k-1} + w_{i,i+2k})\\ 
&\leq\frac{1}{4}(1-2p) +\frac{1}{2}\sum_{i=1}^{N-1}\sum_{k=1}^{\left\lfloor N/2\right\rfloor}2w_{i,i+2k-1}  \\ 
&=\frac{1}{4}(1-2p) + \cut_\zz(\m{G}_X),
\end{align}
where we have used~\eqref{eq:zz_formula},~\eqref{eq:mono_odd_even}, and our definition of the diagonal $w_{i,i}$. We have therefore shown that 
\begin{align}
\cut_\zz(\m{G}_X)&\geq \maxcut(\m{G}_X)-\frac{1}{4}(1-2p). 
\end{align}
Appealing to Theorem~\ref{thrm:graph_formulation}, the proof is concluded.

\section{Additional Results}\label{sec:additional}
\subsection{Suboptimality of General C-partitions}
Recall that $A$ is called a C-partition if the symbols $2k-1$ and $2k$ are on the opposite side of the partition, for any $k$. In the noiseless case, we saw that all the C-partitions (and zigzag among them) are optimal. We now show that in the noisy case, whereas zigzag is almost optimal, non-zigzag C-partitions are in general strictly suboptimal. 

Let us call $\{2k-1,2k\}$ a \emph{$C$-pair} if $2k-1\in A$ and $2k\in \bar{A}$, and a \emph{$\bar{C}$-pair} if $2k-1\in \bar{A}$ and $2k\in A$. We say that $A$ is \emph{zigzag-equivalent} if for any pair of $C$-pair and $\bar{C}$-pair, either all cross-pair edges exist, or none of them exist. It can be directly verified that for non-degenerate $P_X$, if $A$ is zigzag-equivalent then it has the exact same cut weight as the zigzag itself, for essentially the same reasons as in the noiseless case. However, if $A$ is not zigzag-equivalent, then it generally performs strictly worse. 
	\begin{theorem}
		\label{Lemma: ZZ is better than any C-partition}
		Let $A$ be a C-partition. Then
		\begin{align}
		\cut_{\zz}(\m{G}_X) \geq \cut_A(\m{G}_X), 
		\end{align}
                and hence 
                \begin{align}
                \guess_{\zz}(X) \leq \guess_A(X).  
                \end{align}               		
Furthermore, if $P_X$ is non-degenerate, the inequalities are strict if and only if $A$ is not zigzag-equivalent. 
	\end{theorem}
\begin{proof}
We transform all $\bar{C}$-pairs in $A$ into $C$ pairs, thereby transforming the C-partition $A$ into the zigzag partition, and we keep track of the resulting change in the cut weight. To that end, we only need to take into account cross-pair edges, since edges within each pair are in the cut both before and after the transformation. Consider two generic pairs $(1,2)$ and $(3,4)$ that are in a $C$-pair and $\bar{C}$-pair position respectively (they need not be adjacent). Let us write $W_{C\bar{C}}$ and $W_{CC}$ to denote the total weight of edges between the distinct pairs before and after we transform the pair $(3,4)$ into a $C$-pair, respectively. Evoking monotonicity, there are six cases to take into consideration (Figure~\ref{fig:Full edges between C pairs} is helpful here): 
	
		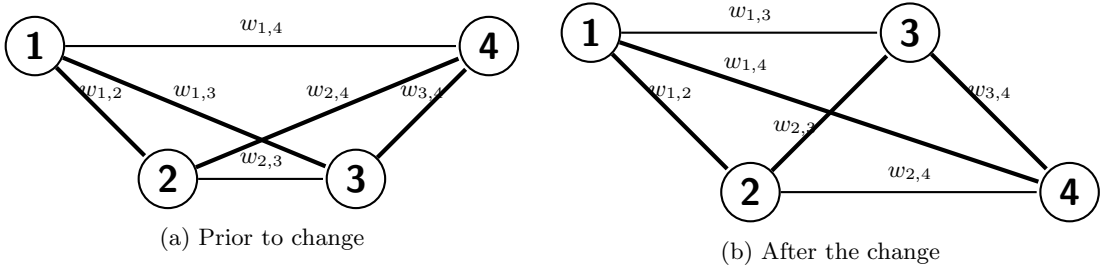
\begin{figure}[H]
			\centering
			\begin{subfigure}{.5\textwidth}
				\centering
				\begin{tikzpicture}
				[-,>=stealth',shorten >=1pt,auto,node distance=2.5cm,thick,main node/.style= {circle,draw,font=\sffamily\Large\bfseries}]
				
				\node[main node] (1) {1};
				\node[main node] (2) [below right of=1] {2};
				\node[main node] (3) [right of=2] {3};
				\node[main node] (4) [above right of=3] {4};
				
				\path[every node/.style={font=\sffamily\small}]
				(1) edge [right, ultra thick] node[pos=0.5,above] {$w_{1,2}$} (2)
				edge [right, ultra thick] node[pos=0.5,above] {$w_{1,3}$} (3)
				edge [right] node[pos=0.5,above] {$w_{1,4}$} (4)
				(2) edge [right] node[pos=0.5,above] {$w_{2,3}$} (3)
				edge [right, ultra thick] node[pos=0.5,above] {$w_{2,4}$} (4)
				(3) edge [right, ultra thick] node[pos=0.5,above] {$w_{3,4}$} (4);
				
				\end{tikzpicture}
				\caption{Prior to change}
				\label{fig:Full edges between C pairs Prior}
			\end{subfigure}\hfill
			\begin{subfigure}{.5\textwidth}
				\centering
				\begin{tikzpicture}
				[-,>=stealth',shorten >=1pt,auto,node distance=3cm,thick,main node/.style= {circle,draw,font=\sffamily\Large\bfseries}]

				\node[main node] (1) {1};
				\node[main node] (2) [below right of=1] {2};
				\node[main node] (3) [above right of=2] {3};
				\node[main node] (4) [below right of=3] {4};
				
				\path[every node/.style={font=\sffamily\small}]
				(1) edge [right, ultra thick] node[pos=0.5,above] {$w_{1,2}$} (2)
				edge [right] node[pos=0.5,above] {$w_{1,3}$} (3)
				edge [left, ultra thick] node[pos=0.3,above] {$w_{1,4}$} (4)
				(2) edge [right, ultra thick] node[pos=0.2,above] {$w_{2,3}$} (3)
				edge [right] node[pos=0.5,above] {$w_{2,4}$} (4)
				(3) edge [right, ultra thick] node[pos=0.5,above] {$w_{3,4}$} (4);
				
				\end{tikzpicture}
				\caption{After the change}
				\label{fig:Full edges between C pairs Posterior}
			\end{subfigure}
			\caption{Edges between C pairs (thick edges are in the cut)}
			\label{fig:Full edges between C pairs}
		\end{figure}

	\begin{enumerate}
        \item Edges between all 4 members of the $C,\bar{C}$-pairs are positive. Since the edges $(1,3)$ and $(2,4)$ are overlapping, then by additivity there is no change in the weight:\label{item:all}
		\begin{align}
		W_{C\bar{C}}&=w_{1,3}+w_{2,4}
		\\&=w_{1,4}+w_{2,3}
		\\&=W_{CC}.
		\end{align}
				
		\item $w_{1,3},w_{2,3},w_{2,4}>0$ and $w_{1,4}=0$: \label{item:1}
		\begin{align}
		W_{C\bar{C}} &=w_{1,3}+w_{2,4}
		\\&=(1-p)p_3-pp_1+(1-p)p_4-pp_2
		\\&=(1-p)p_3-pp_2+(1-p)p_4-pp_1
		\\&=w_{2,3}+(1-p)p_4-pp_1
		\\&\leq w_{2,3}
                \\&=W_{CC}. 
		\end{align}
		The inequality follows since $w_{1,4}=0$ implies $(1-p)p_4-pp_1\leq 0$.

		\item $w_{2,3}>0,w_{2,4}\geq 0$ and $w_{1,3}, w_{1,4}=0$. In this case $W_{C\bar{C}}=w_{2,4}\leq w_{2,3}=W_{CC}$, by monotonicity. \label{item:2}
		\item $w_{2,3}>0,w_{1,3}\geq 0$ and $w_{1,4}, w_{2,4}=0$. In this case $W_{C\bar{C}}=w_{1,3}\leq w_{2,3}=W_{CC}$, by monotonicity.  \label{item:3}
		
		\item $w_{2,3}>0$ and $w_{1,4}, w_{2,4} ,w_{1,3}=0$. In this case $W_{C\bar{C}}=0<w_{2,3}=W_{CC}$. \label{item:4}
		
		\item $w_{1,4}, w_{2,4} ,w_{1,3}, w_{2,3}=0$. In this case $W_{C\bar{C}}=W_{CC}=0$. \label{item:none}

	\end{enumerate}	
We conclude that transforming $A$ into the zigzag partition can only increase the cut weight. Now, it can be directly verified that if $A$ is zigzag-equivalent, which corresponds to cases~\ref{item:all} and~\ref{item:none} above, then all the inequalities hold with equality, essentially since in this case additivity holds. Conversely, and assuming non-degeneracy, if $A$ is not zigzag-equivalent there must exist a $C$-pair and a $\bar{C}$-pair such that some but not all of the cross-pair edges exist. This situation falls into cases~\ref{item:1},~\ref{item:2},~\ref{item:3} or~\ref{item:4} above. In all those cases, it can again be directly verified that the inequalities are strict under the non-degeneracy assumption. 
\end{proof}

\subsection{Cutoff Rate with Feedback}\label{subsec:cutoff}
Consider the case of communicating a message $W\sim \mathrm{Uniform}(\{1,\ldots,M\})$ over $n$ independent uses of a Binary Symmetric Channel (BSC) with crossover probability $p$. A {\em coding scheme} is a mapping $e:\{1,\ldots,M\}\to \{0,1\}^n$ from the message set to the channel inputs. Let $Y^n$ denote the channel outputs corresponding to the input $e(W)$. The \emph{cutoff rate $R_{\mathrm{cutoff}}(p)$} associated with the channel is the supremum over all rates $R$ for which there exists a sequence (in $n$) of coding scheme with $M=2^{nR}$, such that given $Y^n$, the expected number of messages that are more likely than the correct message tends to zero, or equivalently, the minimal conditional guessing time tends to one: 
\begin{align}
	\lim_{n\to\infty}\guess(W\mid Y^n) = 1. 
\end{align}
It is well known~\cite{arikan1996inequality} that 
\begin{align}
	R_{\textrm{cutoff}}(p) = 1-h_{\frac{1}{2}}(p),
\end{align}
where $h_{\frac{1}{2}}(p)=2\log\left(\sqrt{p}+\sqrt{1-p}\right)$ is the binary R\'{e}nyi entropy of order $1/2$. Arikan~\cite{arikan1988upper,arikan1996inequality} has characterized the cutoff rate for general discrete memoryless channels, and also for general moments of the guessing time. Here we restrict our discussion to the BSC and the first moment. 

Now assume that the channel is equipped with instantaneous noiseless feedback from the receiver back to the transmitter. A {\em feedback coding scheme} is a sequence of mappings $e_k:\{1,\ldots,M\}\times \{0,1\}^{k-1}\to \{0,1\}$, such that the input to the channel at time $k$ is $e_k(W,Y^{k-1})$. The \emph{cutoff rate with feedback $R_{\mathrm{cutoff,fb}}(p)$} is defined similarly to its no-feedback counterpart, where now feedback coding schemes are allowed. It was conjecture by Arikan~\cite{arikan1988upper} and recently shown by Bunte and Lapidoth~\cite{bunte2013average,bunte2014listsize} that $R_{\mathrm{cutoff,fb}}(p)=R_{\mathrm{cutoff}}(p)$, i.e., that feedback does not increase the cutoff rate (for general discrete memoryless channels). This was achieved by proving that 
\begin{align}
	\guess(W\mid Y^n) \geq \frac{2^{n(R-R_{\mathrm{cutoff}}(p))}}{1+nR}. 
\end{align} 
for any feedback coding scheme of rate $R$. Thus in fact, the definition of the cutoff rate (either with or without feedback) can be slightly relaxed; it is the maximal rate $R$ such that the message guessing can be made sub-exponential, i.e., such that 
\begin{align}
\lim_{n\to\infty}\frac{1}{n}\log\guess(W\mid Y^n) = 0,  
\end{align}
 
In this section, we provide an explicit feedback communication scheme that achieves the cutoff rate in this slightly relaxed sense. Not surprisingly, our scheme is to simply use the zigzag partition repeatedly; namely, at each step $k$ the transmitter partitions the current posterior distribution of $W$ given $Y^{k-1}$ using the zigzag rule, and generates the channel input $X_k$ as the indicator pertaining to the side of the partition where the message $W$ lies. This can be thought of as the ``guessing analogue'' of the probabilistic bisection scheme of Horstein~\cite{horstein1963sequential}, which uses a median partition of the posterior, and was later shown to achieve channel capacity within the posterior matching framework~\cite{shayevitz2011optimal}. We prove the following. 
\begin{theorem}\label{thrm:zz_cutoff}
For any rate $R < R_{\textrm{cutoff,fb}}(p)$, using the zigzag partition scheme described above yields a guessing time of $\guess(W\mid Y^n) = O(n^2)$. 
\end{theorem}
\begin{remark}
  If one could claim that using the zigzag partition in each step is optimal, then that would also recover the converse of~\cite{bunte2014listsize}. However, unlike the noiseless case (see Example~\ref{ex:guessing_ent}), that is not immediately clear. This multi-step problem can be thought of as a finite-horizon Markov decision process, and it does not appear trivial to show that being ``greedy''  in each step (i.e., using the zigzag) is the best overall strategy. Nevertheless, Theorem~\ref{thrm:zz_cutoff} and the converse in~\cite{bunte2014listsize} show that this scheme is  optimal at least in the exponential sense. 
\end{remark}
\begin{proof}
For the meanwhile, let us ignore integer issues and assume that the $M$ symbols are infinitely divisible, i.e., that we can chop up a symbol into any number of parts and place them on opposite sides of the partition. Since $W$ is uniform, the graph $\m{G}_0$ associated with it (prior to transmission) is a complete graph. We can therefore use an arbitrary equal partition for the first transmission. After the first transmission and given the first channel output $Y_1$, we are left with a posterior distribution whose graph $\m{G}_1$ is a disjoint union of two cliques, each of size $M/2$, where disjointness is a result of the ratio between symbols that have received a ``yes'' and those that received a ``no'', which is exactly $\frac{1-p}{p}$. The probabilities of all members in the first (resp. second) clique are hence all equal to $\tfrac{2p}{M}$ (resp. $\tfrac{2(1-p)}{M}$). The zigzag strategy at this point is thus to partition each of the cliques equally, which will result in $\m{G}_2$ having three cliques of sizes $\{M/4,M/2,M/4\}$, with probabilities in each clique equal to $\tfrac{4}{M}\cdot \{p^2, p(1-p), (1-p)^2\}$ respectively. Continuing this binomial process, we can see that after $k$ channel uses the associated graph $\m{G}_k$ is a disjoint union of $k+1$ cliques of sizes $\{m_j^{(k)} = M\cdot 2^{-k}{k \choose j}\}_{j=0}^k$. The probability of each member in the $j$th cliques is $q_j^{(k)}=\frac{1}{M}2^{k}p^{j}(1-p)^{k-j}$. Note that this structure holds regardless of the specific channel output sequence $Y^k$, which only determines the specific members of each clique but not their sizes or distribution. 

Let us now resolve the issue of non-integer partitions. Using the actual zigzag strategy, whenever a clique is of odd size we have a single ``leftover symbol'' that gets assigned to an arbitrary side of the partition, rather than being ``split in half''  as we fictitiously assumed above. Since at each step $k$ there are exactly $k+1$ cliques, the number of such leftover symbols during the entire process is at most $O(n^2)$. Since the total weight of edges connected to a single vertex in any graph of the form we are considering is at most one, the total weight associated with the leftover symbols is also $O(n^2)$. Hence, computing the total ``cut weight'' induced by all the non-integer partitions is $O(n^2)$-close to the cut weight induced by the zigzag partitions strategy. 

Following the above, let us estimate the minimal expected guessing time following $n$ steps of the zigzag strategy, by summing up all the intermediate cut weights:
\begin{align}
\guess(W\mid Y^n)&= \guess(W) - \frac{1-2p}{4}\sum_{k=0}^{n-1}\sum_{j=0}^k m_j^2 q_j  + O(n^2)\\ 
&=\frac{M+1}{2} - \frac{1-2p}{4}\sum_{k=0}^{n-1}\sum_{j=0}^k m_j^2 q_j  + O(n^2)\\
&=\frac{M+1}{2} - \frac{(1-2p)}{4}\sum_{k=0}^{n-1}M\cdot 2^{-k}\sum_{j=0}^k {k \choose j}^2p^j(1-p)^{k-j}  + O(n^2)\\
&= 2^{nR-1}\left(1-\left(\frac{1}{2}-p\right) \sum_{k=0}^{n-1}\left(\frac{1-p}{2}\right)^k \sum_{j=0}^k {k \choose j}^2 \left(\frac{p}{1-p}\right) ^{j}\right)  + O(n^2)\\
& = 2^{nR-1}(1-B_n) + O(n^2), \label{eq:guess_final}
\end{align}
where $B_n$ was implicitly defined. To proceed, let us recall two well known properties of the Legendre polynomials. 
\begin{lemma}
Let $P_k(x)$ be the $k$th-degree Legendre polynomial. Then for any $0 <  \alpha < 1$
\begin{align}
\sum_{j=0}^k{k \choose j}^2\alpha^j = (1-\alpha)^kP_k\left(\frac{1+\alpha}{1-\alpha}\right)  
\end{align}  
and 
\begin{align}
\sum_{k=0}^\infty\alpha^kP_k(x) = \frac{1}{\sqrt{1+\alpha^2-2\alpha x}}.
\end{align}
\end{lemma}
Using these relations, we have
\begin{align}
B_\infty &= \left(\frac{1}{2}-p\right) \sum_{k=0}^{\infty}\left(\frac{1-p}{2}\right)^k \sum_{j=0}^k {k \choose j}^2 \left(\frac{p}{1-p}\right) ^{j} \\ 
&= \left(\frac{1}{2}-p\right)\sum_{k=0}^\infty \left(\frac{1}{2}-p\right)^kP_k\left(\frac{1}{1-2p}\right)\\
&= \left(\frac{1}{2}-p\right)\frac{1}{\sqrt{1+(1/2-p)^2-2(1/2-p)/(1-2p)}}\\
&=1.
\end{align}
With this in hand, let us define 
\begin{align}
  \beta^*(p) \dfn \max_{\alpha\in[0,1]}2h(\alpha)+\alpha\log\frac{p}{1-p},
\end{align}
where $h(x)=-x\log{x}-(1-x)\log(1-x)$ is the binary entropy function. Note that the maximizing value is $\alpha^* = \frac{\sqrt{p}}{\sqrt{p}+\sqrt{1-p}}$. We can now rewrite~\eqref{eq:guess_final} as follows (using $j=\alpha k$):
\begin{align}
  \guess(W\mid Y^n) &= 2^{nR-1}(B_\infty-B_n) + O(n^2) \\ 
& = 2^{nR-1}\left(\frac{1}{2}-p\right) \sum_{k=n}^{\infty}\left(\frac{1-p}{2}\right)^k \sum_{j=0}^k {k \choose j}^2 \left(\frac{p}{1-p}\right) ^{j}  + O(n^2)\\
&\leq  2^{nR-1}\left(\frac{1}{2}-p\right)\sum_{k=n}^{\infty}\left(\frac{1-p}{2}\right)^k (k+1)\,2^{k\beta^*(p)} + O(n^2)\\
&=  2^{n(R+\beta^*(p))-1}\left(\frac{1}{2}-p\right)\sum_{k=0}^{\infty} (n+k+1)\,2^{k\left(\log\frac{1-p}{2}+\beta^*(p)\right)} + O(n^2)\\
&= 2^{n(R+\log\frac{1-p}{2}+\beta^*(p))}\cdot O(1) + O(n^2). 
\end{align}
We have used the standard binary entropy bound ${k \choose j}\leq 2^{kh(j/k)}$ for the binomial coefficients~\cite{cover2012elements}, and the last equality holds since $\log\frac{1-p}{2}+\beta^*(p) < 0$, which is verified below in~\eqref{eq:hren1}-\eqref{eq:hren2}. Thus, we conclude that using the zigzag strategy yields $\guess(W\mid Y^n) = O(n^2)$ whenever the rate satisfies
\begin{align}\label{eq:hren1}
  R &< -\left(\log\frac{1-p}{2}+\beta^*(p)\right)\\
& = 1-\log(1-p) -\frac{\sqrt{p}}{\sqrt{p}+\sqrt{1-p}}\log\frac{p}{1-p} - 2h\left(\frac{\sqrt{p}}{\sqrt{p}+\sqrt{1-p}}\right) \\
& = 1-h_{\frac{1}{2}}(p),\label{eq:hren2} 
\end{align}
where the last equality follows by direct computation. This concludes the proof. 
\end{proof}

\section{Discussion}
We have shown that the zigzag partition, which amounts to querying whether $X$ has an odd or even index when ordered in descending order of probabilities, is the best question to ask in terms of reducing the minimal expected guessing time, up to a constant of at most $(1-2p)/4$, where $p$ is the probability of getting a wrong answer. This small constant is generally an artifact of our proof technique, since for $p=0$ we know that the zigzag partition is precisely optimal. Moreover, as stated in Conjecture~\ref{conj}, we believe that the zigzag partition is precisely optimal for any distribution and any error probability $p$.  Our paper only deals with the expectation of the guessing time; it may be desirable to explore partitions that minimize other moments of the guessing time as well. Our graph-theoretic approach relies heavily on the linearity of expectation and does not seem to naturally lend itself to this problem, hence it is likely that a different approach would be needed. Additionally, our treatment has been limited to binary questions and a binary symmetric channel to the Oracle. It is interesting to study the structure of optimal partitions in more general cases, where either the questions are multiple-choice, the channel is asymmetric, or both. 

\section{Acknowledgments} 
The authors are grateful to the reviewers and the associate editor for their many comments and suggestions that have significantly improved the presentation of the paper. 

\appendix 

\section{Appendix}
\begin{proof}[Proof of Proposition~\ref{lem:guessing_time_prop}]
We first note that $G(X\mid Y)\leq G(X)$, a fact that can be verified directly from definition or observed as a special case of~\cite[Corollary 1]{arikan1996inequality}. Thus, the function $\guess(X)$ is concave in the distribution $P_X$. Furthermore, $\guess(X)$ is also permutation invariant by definition. Thus, $\guess(X)$ is Schur-concave. Since $P_{X'}$ is obtained from $P_X$ by multiplying by a doubly-stochastic matrix, then $P_{X'}$ is majorized by $P_X$. The result follows since Schur-concave functions respect that majorization partial order~\cite{bhatia2013matrix}. 
\end{proof}

\begin{proof}[Proof of Corollary~\ref{cor:guess_simple_bounds}]
A deterministic distribution strictly majorizes any other non-deterministic distribution, and the uniform distribution is strictly majorized by any other distribution. The result now follows immediately from Proposition~\ref{lem:guessing_time_prop} and a direct computation. 
\end{proof}

\begin{proof}[Proof of Lemma~\ref{lem:non-degenerate}]
Let $Y_A = \ind(X\in A)\oplus V$. Applying Bayes law, the posterior distribution of $X$ given $Y_A$ is given by 
\begin{align}
P_{X|Y_A}(i\mid y) = \frac{1}{P_{Y_A}(y)} \cdot\left\{\begin{array}{lr} (1-p)p_i & y = \ind(i\in A) 
\\ pp_i & y = \ind(i\not\in A)\end{array}\right..
\end{align}
Note that $P_{Y_A}(y) \neq 0$ for any $0<p<1$ and $y\in\{0,1\}$. We can see that $P_{X|Y_A}(\cdot \mid y)$ has all distinct symbol probabilities for a fixed $A, y$ and $0<p<1$ if and only if $p_i\neq p_j$ whenever $i$ and $j$ are on the same side of the partition $A$, and $\frac{p_i}{p_j}\not\in\{\frac{p}{1-p}, \frac{1-p}{p}\}$ whenever $i,j$ are on opposite sides of the partition $A$. Thus, $P_X$ is non-degenerate if and only if $\frac{p_i}{p_j}\not\in\{1,\frac{p}{1-p}, \frac{1-p}{p}\}$ for all $i,j\in\{1,\ldots,N\}$ with $i\neq j$. 

Now, let $P_X$ be any distribution over $\{1,\ldots,N\}$. Fix some $\delta > 0$ and draw a $W^n \stackrel{\mathrm{i.i.d}}{\sim}\mathrm{Uniform}([0,\delta])$. Define the (random) distribution $Q_X$ to have symbol probabilities
\begin{align}
	q_i \dfn  \frac{p_i + W_i}{1+\sum_{j=1}^NW_j}.
\end{align}  
First, note that 
\begin{align}
	|p_i-q_i| & \leq \left|\frac{p_i\sum_{j=1}^NW_j - W_i}{1+\sum_{j=1}^NW_j}\right|\\
	& \leq \delta (N+1).
\end{align}
with probability one. Setting any $\delta < \frac{\varepsilon}{N(N+1)}$, we have that 
\begin{align}
	\Pr(\|P_X-Q_X\|_1  \geq \varepsilon) = 0. 
\end{align}
Now, the probability ratio
\begin{align}
	\frac{q_i}{q_j} = \frac{p_i+W_i}{p_j+W_j}
\end{align}
has a continuous p.d.f with support $\left[\frac{p_i}{p_j+\delta}, \frac{p_i+\delta}{p_j}\right]$. Hence, 
\begin{align}
	\Pr\left(\frac{q_i}{q_j} = c\right) =0
\end{align}
for any constant $c$. Applying the union bound, we see that $Q_X$ is non-degenerate and $\|P_X-Q_X\|_1 < \varepsilon$, with probability one. In particular, there exists a deterministic choice of such a distribution. 
\end{proof}

\begin{proof}[Proof of Lemma~\ref{lem:eqv_psd}]
  $\mathbf{x}^T\mathbf{C}\mathbf{x} \geq 0$ if and only if $(\mathbf{A}\mathbf{x})^T\mathbf{B}(\mathbf{A}\mathbf{x}) \geq 0$.
\end{proof}

\begin{proof}[Proof of Lemma~\ref{lem:gdisks}]
	This result is classical; we give a short proof for completeness. Let $\lambda$ be an eigenvalue of $\mathbf{B}$ and let $\mathbf{x}$ be a corresponding eigenvector, and let $i=\argmax_j| x_j|$. Then $| x_i| > 0$ unless $\mathbf{x}$ is the all-zeros vector, and $|\frac{\mathbf{x}_j}{\mathbf{x}_i}|\leq 1$ for any $j$.
	Since $\mathbf{Bx}=\lambda\mathbf{x}$, we have 
	\begin{align}
	\sum_{j\neq i}b_{i,j}x_j=\lambda x_i-b_{i,i}x_i,	
	\end{align}
and therefore 
\begin{align}
|\lambda-b_{i,i}| &= \big{|}\sum_{j\neq i}\frac{b_{i,j}x_j}{x_i}\big{|}\\
&\leq \big{|}\sum_{j\neq i}b_{i,j}\big{|}\\
&\leq \sum_{j\neq i}\left| b_{i,j}\right|\\
&=R_i. 
\end{align}
\end{proof}

\bibliographystyle{IEEEtran}
\bibliography{GuessingPaper}

\end{document}